\documentclass[11pt]{amsart}

\usepackage[a4paper,margin=1in]{geometry}
\usepackage{amssymb}
\usepackage[foot]{amsaddr}
\usepackage{setspace}
\usepackage{xcolor}
\definecolor{Navy}{RGB}{30, 70, 235}
\definecolor{Green}{RGB}{0, 128, 0}
\IfFormatAtLeastTF{2026-06-01}{}{\usepackage{aliascnt}}
\usepackage[bookmarks=true,hypertexnames=false,pagebackref]{hyperref}
\hypersetup{colorlinks=true, citecolor=Green, linkcolor=Navy, urlcolor=blue}
\usepackage{prettyref}
\usepackage{todonotes}

\newtheorem{theorem}{Theorem}

\IfFormatAtLeastTF{2026-06-01}{
  \newtheorem{definition}[theorem]{Definition}
  \newtheorem{proposition}[theorem]{Proposition}
  \newtheorem{lemma}[theorem]{Lemma}
  \newtheorem{corollary}[theorem]{Corollary}
  \newtheorem{observation}[theorem]{Observation}
  \newtheorem{claim}[theorem]{Claim}
}{
  \newaliascnt{definition}{theorem}
  \newtheorem{definition}[definition]{Definition}
  \aliascntresetthe{definition}

  \newaliascnt{proposition}{theorem}
  \newtheorem{proposition}[proposition]{Proposition}
  \aliascntresetthe{proposition}

  \newaliascnt{lemma}{theorem}
  \newtheorem{lemma}[lemma]{Lemma}
  \aliascntresetthe{lemma}

  \newaliascnt{corollary}{theorem}
  \newtheorem{corollary}[corollary]{Corollary}
  \aliascntresetthe{corollary}

  \newaliascnt{observation}{theorem}
  
  \aliascntresetthe{observation}

  \newaliascnt{claim}{theorem}
  
  \aliascntresetthe{claim}
}

\theoremstyle{remark}
\IfFormatAtLeastTF{2026-06-01}{
  \newtheorem{remark}[theorem]{Remark}
}{
  \newaliascnt{remark}{theorem}
  
  \aliascntresetthe{remark}
}
\newtheorem*{remark*}{Remark}

\makeatletter
\renewcommand{\prettyref}[1]{\expandafter\pr@dispatch#1:\@nil}
\def\pr@dispatch#1:#2:\@nil{%
  \@ifundefined{pr@#1}{%
    \PackageWarning{prettyref}{Reference format #1\space undefined}%
    \ref{#1:#2}%
  }{%
    \@nameuse{pr@#1}{#1:#2}%
  }%
}
\makeatother

\newrefformat{thm}{\hyperref[#1]{Theorem~\ref*{#1}}}
\newrefformat{prop}{\hyperref[#1]{Proposition~\ref*{#1}}}
\newrefformat{lem}{\hyperref[#1]{Lemma~\ref*{#1}}}
\newrefformat{cor}{\hyperref[#1]{Corollary~\ref*{#1}}}
\newrefformat{claim}{\hyperref[#1]{Claim~\ref*{#1}}}
\newrefformat{def}{\hyperref[#1]{Definition~\ref*{#1}}}
\newrefformat{sec}{\hyperref[#1]{Section~\ref*{#1}}}
\newrefformat{subsec}{\hyperref[#1]{Section~\ref*{#1}}}
\newrefformat{eq}{\hyperref[#1]{Equation~\textup{(\ref*{#1})}}}
\newrefformat{eqn}{\hyperref[#1]{Equation~\textup{(\ref*{#1})}}}
\newcommand{\Cref}[1]{\prettyref{#1}}

\newcommand{\Ex}{\mathop{\mathbb{{}E}}\nolimits}
\renewcommand{\Pr}{\mathop{\mathrm{Pr}}\nolimits}
\newcommand{\abs}[1]{\ensuremath{\left\vert#1\right\vert}}

\def\TV{d_{\mathrm{TV}}}
\def\Tmix{T_{\mathrm{mix}}}
\def\mono{\operatorname{mono}}
\def\Lift{\operatorname{Lift}}
\def\eps{\varepsilon}
\def\defeq{:=}
\newcommand{\NP}{\textnormal{\textbf{NP}}}

\title[An $\widetilde{O}(n^2)$-Time Sampler for Zero-Field Ferromagnetic Ising Models]{An $\widetilde{O}\left(n^2 \right)$-Time Sampler for Zero-Field Ferromagnetic Ising Models}
\author{Weiming Feng}
\address[Weiming Feng]{School of Computing and Data Science, The University of Hong Kong, Hong Kong, China. Email: \texttt{wfeng@hku.hk}. Weiming Feng acknowledges the support of ECS grant 27202725 from Hong Kong RGC.}
\author{Heng Guo}
\address[Heng Guo]{School of Informatics, University of Edinburgh, Edinburgh, United Kingdom. Email: \texttt{hguo@inf.ed.ac.uk}.}
\author{Yiyao Zhang}
\address[Yiyao Zhang]{State Key Laboratory for Novel Software Technology, New Cornerstone Science Laboratory, Nanjing University, Nanjing, China. Email: \texttt{zhangyiyao@smail.nju.edu.cn}.}
\date{}

\begin{document}

\begin{abstract}
  We give an approximate sampler for ferromagnetic Ising models with no field on arbitrary graphs that runs in time $\widetilde O(m+n)+\widetilde O_\beta(n^2\log^2 (1 / \eps))$, where $n$ and $m$ are the numbers of vertices and edges, respectively, and $\eps$ is the approximation error.
  Our approach combines Bencz\'ur--Karger cut sparsification with a new mixing time analysis of the Glauber dynamics for the random-cluster model. 
  The mixing time analysis features a new monotone edge-count Poincar\'e inequality.
\end{abstract}

\maketitle

\section{Introduction}

The Ising model is perhaps the most studied and most well-known statistical physics model.
It features both a simple description and a rich mathematical structure.
The main computational task for the Ising model is to sample from its Gibbs distribution, which we define next.

Let $G:=(V,E)$ be a connected undirected (multi)graph.
Let $\beta>0$ be a parameter describing the interaction strength.
For any configuration $\sigma\in\{0,1\}^V$, define the number of monochromatic edges as $\mono(\sigma):=|\{uv\in E:\sigma_u=\sigma_v\}|$, where parallel edges are counted with multiplicities.
Let $Z^{\mathrm{Ising}}_\beta:=\sum_{\sigma\in\{0,1\}^V}\beta^{\mono(\sigma)}$ be the so-called partition function.
The \emph{Gibbs distribution} of the zero-field Ising model on $G$ with interaction parameter $\beta$ is defined by
\begin{align*}
  \forall \sigma\in\{0,1\}^V,\quad
  \mu^{\mathrm{Ising}}_{\beta,G}(\sigma)
  :=\frac{\beta^{\mono(\sigma)}}{Z^{\mathrm{Ising}}_\beta}.
\end{align*}
When $\beta>1$, the system is called ferromagnetic, and anti-ferromagnetic if $0 < \beta < 1$.

The computational complexity of sampling from $\mu^{\mathrm{Ising}}_{\beta,G}$ has been a classic topic in theoretical computer science.
The pioneering work of Jerrum and Sinclair \cite{JS93} showed that a polynomial-time approximate sampler\footnote{The original paper \cite{JS93} does not give a sampler explicitly, but it is not hard to obtain one, as shown by Randall and Wilson \cite{RW99}.} exists for ferromagnetic Ising models, even with consistent external fields,\footnote{Roughly speaking, external fields are vertex weights. In this paper we focus on the zero-field case, and thus do not introduce them formally.} whereas the problem becomes \NP{}-hard for anti-ferromagnetic Ising models.

For anti-ferromagnetic Ising models, we have by now fully understood the computational phase transition regarding the parameter $\beta$ and the maximum degree \cite{LLY13,SST14,SS14,GSV16,CJMYZ26}.
In the uniqueness regime, the running time of the sampler has been improved to near-linear \cite{CLV23,AJKPV22,CFYZ22}.
However, on the ferromagnetic side, algorithmic progress is relatively sparse.
The only running time improvement is obtained for the non-zero field case, where Chen and Zhang \cite{CZ23} introduced a near-linear time sampling algorithm.
For the zero-field case, prior to this work, the fastest sampler was still the Jerrum--Sinclair algorithm \cite{JS93}, which runs in time $O(m^2n^4(m+\log (1 / \eps)))$, where $n$ denotes the number of vertices, $m$ the number of edges, and $\eps$ the error.
This running time can be as high as $O(n^{10})$ for dense graphs.
In this paper, we bring it down to $\widetilde{O}(n^2)$ in simple graphs.
For two probability distributions $\mu$ and $\nu$ on the same finite state space $\Omega$, define their \emph{total variation distance} by $\TV(\mu,\nu):=\frac{1}{2}\sum_{x\in\Omega}|\mu(x)-\nu(x)|$.

Our main theorem is stated as follows. 
Here, $G$ is allowed to be a multigraph.

\begin{theorem}\label{thm:main}
For any constant $\beta>1$, there is an algorithm such that, given $G=(V,E)$ with $n$ vertices and $m$ edges, and an error bound {$\eps\in(0,1/2)$}, it outputs a random configuration with law $\nu$ satisfying $\TV(\nu,\mu^{\mathrm{Ising}}_{\beta,G})\leq\eps$ in time
\[
  O(m\log^2 n+n)+O_\beta\bigl(n^2\log^2 n\,\log^2(n/\eps)\bigr).
\]
In particular, if $G$ is simple, the total running time is $O_\beta(n^2\log^2 n\,\log^2(n/\eps))$.
\end{theorem}

\subsection{Algorithm Overview}
The two main ingredients in the algorithm of \Cref{thm:main} are sparsification and an improved mixing time analysis for the random-cluster model.
Although the authors provided some high-level guidance, the detailed proofs of both ingredients were found by LLMs.
Subsequently, the authors simplified, streamlined, and wrote all of the proofs.
The authors also take full responsibility for the correctness of the paper.
For consistency with standard mathematical exposition, the words ``we'' and ``our'' are used throughout the paper, including when presenting arguments originating from the model.

The first ingredient of the algorithm is a sparsification argument.
It states that the original Gibbs distribution can be accurately approximated, in total variation distance, by another Gibbs distribution defined on a suitably sparsified version of the underlying graph.

\begin{theorem}[Ising sparsification]
  \label{thm:Ising:sparsification}
  For every constant $\beta>1$, there is an algorithm that, given $G:=(V,E)$ and an error bound $\eta\in(0,1)$, explicitly constructs a sparsified multigraph $G':=(V',E')$ and an implicit representation of a lifting map $\Lift:\{0,1\}^{V'}\to\{0,1\}^V$ in time $O(m\log^2 n+n)$, such that $|V'|\leq |V|=n$ and $|E'|=O_\beta(n\log(n/\eta))$, where edges are counted with multiplicity. Moreover:
  \begin{itemize}
    \item If $\mu'$ is the zero-field Ising Gibbs distribution on $G'$ with interaction parameter $\beta$, then $\TV(\Lift(\mu'),\mu^{\mathrm{Ising}}_{\beta,G})\leq\eta$, where $\Lift(\mu')$ denotes the law of $\Lift(X)$ for $X\sim\mu'$.
    \item Given any $\sigma\in\{0,1\}^{V'}$, the configuration $\Lift(\sigma)$ can be computed in $O(n)$ time.
  \end{itemize}
\end{theorem}

The main technique we use is the machinery behind the classic cut sparsifier by Bencz\'ur and Karger \cite{BK15} and Karger's cut counting bounds \cite{Karger93,Karger99}.
Indeed, the Ising partition function can be viewed as a generating function for cuts, and it is perhaps not a surprise that these tools are useful here.
However, this argument appears to be new in the context of partition function approximation, and is the first sparsification argument in this context as far as we are aware.
In fact, we notice that the same technique works for all-terminal reliability \cite{GJ19} as well, which improves the running time of the FPRAS for all-terminal reliability from $\widetilde{O}(mn/\eps^2)$ \cite{CGZZ24} to $\widetilde{O}(m+n^2/\eps^2)$.
Details are given in \Cref{sec:reliability}.

The intuition here is simple.
For any subgraph with sufficiently high edge connectivity, with high probability, all vertices have the same spin, and thus can be contracted while approximating the original Gibbs distribution.
This crucially relies on the fact that $\beta$ is a constant.
If $\beta=1+o(1)$, such as $\beta=1+c/n$, then the argument is no longer valid.
The machinery in \cite{BK15} allows us to correctly and efficiently identify what subgraphs need to be contracted.
With this sparsification, we only need to consider Ising models on graphs with $m = O_\beta(n\log(n/\eps))$ edges.

On the sparsified graph, we use Glauber dynamics for the random-cluster model.
This model was introduced by Fortuin and Kasteleyn \cite{FK72}, which can be coupled with the ferromagnetic Ising model.
It is defined as follows.
Again, suppose the graph is $G=(V,E)$.
For $S\subseteq E$, let $\kappa(S)$ denote the number of connected components of $(V,S)$, including isolated vertices. Given $p\in(0,1)$, the random-cluster model on $G$ defines the probability distribution
\begin{align*}
  \forall S\subseteq E, \quad
  \mu^{\mathrm{RC}}_{p,G}(S)
  :=\frac{p^{|S|}(1-p)^{m-|S|}2^{\kappa(S)}}{Z_{\mathrm{RC}}},
\end{align*}
where $Z_{\mathrm{RC}}$ is the corresponding partition function. Edges in $S$ are called open, and edges in $E\setminus S$ are called closed.
Parallel edges are treated as distinct edges.
Edwards and Sokal \cite{ES88} showed that the ferromagnetic Ising model and the random-cluster model can be tightly coupled, which is stated in the following proposition.
\begin{proposition}[Edwards--Sokal coupling~\cite{ES88}]
  \label{prop:Edwards:Sokal}
  Let $p:=1-1/\beta$. Sample $S\sim\mu^{\mathrm{RC}}_{p,G}$ and assign an independent uniformly random spin in $\{0,1\}$ to each connected component of $(V,S)$, giving every vertex the spin of its component. The resulting configuration has distribution $\mu^{\mathrm{Ising}}_{\beta,G}$.
  Conversely, sample $\sigma\sim\mu^{\mathrm{Ising}}_{\beta,G}$ and, conditional on $\sigma$, include each monochromatic edge in $S$ independently with probability $p$ and close every bichromatic edge. Then $S\sim\mu^{\mathrm{RC}}_{p,G}$.
\end{proposition}

Thus, it suffices to sample from $\mu^{\mathrm{RC}}_{p,G}$.
For this, we use the single-edge heat-bath \emph{Glauber dynamics}.
Start from an arbitrary subset $X_0\subseteq E$.
At each step, given the current configuration $X$, choose an edge $e\in E$ uniformly at random and let $u$ and $v$ be its endpoints. The Glauber dynamics resamples the state of $e$ conditional on $X\setminus\{e\}$. Formally, the next configuration is $X\cup\{e\}$ with probability
\begin{align*}
  p' := \begin{cases}
    p, &
      \text{if $u$ and $v$ lie in the same connected component of $(V,X\setminus\{e\})$,}\\
    \displaystyle\frac{p}{p+2(1-p)}
      & \text{otherwise}.
  \end{cases}
\end{align*}
Otherwise, with probability $1-p'$, the next configuration is $X\setminus\{e\}$.  Parallel edges are treated as distinct edges.
It is well known that the Glauber dynamics converges to the stationary distribution $\mu^{\mathrm{RC}}_{p,G}$.

Let $P$ denote the transition kernel of the Glauber dynamics above.  For $\eps\in(0,1)$, define its mixing time by
\begin{align*}
  \Tmix(P,\eps)
  :=\min\left\{t\geq0:
    \max_{X_0\subseteq E}
    \TV\bigl(P^t(X_0,\cdot),\mu^{\mathrm{RC}}_{p,G}\bigr)
    \leq\eps
  \right\}.
\end{align*}
We prove the following mixing time bound for the Glauber dynamics.
\begin{theorem}[Mixing time of the Glauber dynamics]
  \label{thm:RC:Glauber:mixing}
Let $p \in (0,1)$ be a constant.
For every connected loopless multigraph $G:=(V,E)$ with $n:=|V|$ and $m:=|E|>0$ and every $\eps\in(0,1)$,
\[
\Tmix(P,\eps) \leq \left\lceil 2nm \log \frac{2m}{\eps} \right\rceil.
\]
\end{theorem}

Previously, the best (and only) mixing-time upper bound for this dynamics was the result of Guo and Jerrum \cite{GJ18}, which is an $O(m^2n^4(m+\log (1 / \eps)))$ bound.

Our proof of \Cref{thm:RC:Glauber:mixing} uses a monotone grand coupling.
The key is to control how fast two chains from the all-closed and all-open configurations coalesce.
We show that the difference of their expected sizes contracts geometrically at a rate of roughly $1-\frac{1}{mn}$,
which leads to the bound in \Cref{thm:RC:Glauber:mixing}.
This decay follows from a new inequality, which we call the monotone edge-count Poincar\'e inequality.
A standard Poincar\'e inequality bounds the variance in terms of the Dirichlet form for a single function,
whereas our version considers the covariance of two functions against their Dirichlet form.
More precisely, we restrict one function to be monotone, and fix the other to be the edge count observable.
With an appropriate choice of the monotone function,
the covariance becomes the expected size of the difference at time $t$,
and the Dirichlet form equals the one-step decrease in this expected size from time $t$ to $t+1$.
Thus, such an inequality establishes a geometrical decay of this difference.

To prove the inequality, we use the Grimmett--Janson coupling \cite{GJ09} to sample a weighted even subgraph and then lift it to a random-cluster configuration.
The law of total covariance separates the lift and even-subgraph contributions.
The lift part consists of independent Bernoulli additions, and is relatively easy to handle.
The even-subgraph contribution is reduced to local differences by flipping edges along the cycles in the symmetric difference.
These edge flips resemble previous constructions of canonical paths \cite{JS93} or multi-commodity flows \cite{GJ18}.
The key difference is that our argument is much more direct and does not rely on bounding the congestion in any way.
It completely bypasses several lossy steps in previous arguments.

We expect the mixing time bound in \Cref{thm:RC:Glauber:mixing} not to be tight, but we also conjecture that the tight bound is not near-linear.
We discuss this further in \Cref{sec:conclude}.

Each transition of the Glauber dynamics on the random-cluster model can be implemented in $O(\log^2 n)$ amortized time using a fully dynamic connectivity data structure~\cite{WN13,CZ23}.
Therefore, one can use Glauber dynamics to sample a random configuration from $\mu^{\mathrm{RC}}_{p,G}$ in time $\widetilde O_p(nm \log (1 / \eps))$.
The algorithm in \Cref{thm:main} is obtained by combining the Ising sparsification theorem in \Cref{thm:Ising:sparsification}, the Edwards--Sokal coupling in \Cref{prop:Edwards:Sokal}, and the mixing-time bound for the Glauber dynamics in \Cref{thm:RC:Glauber:mixing}.
The full algorithm is given in \Cref{sec:proof:main}.

\section{A Near-Linear-Time Ising Sparsification}\label{sec:sparsification}

In this section, we prove \Cref{thm:Ising:sparsification}.
Self-loops contribute the same factor $\beta$ to every spin configuration and therefore do not affect the Ising Gibbs distribution.  We may delete all self-loops and assume throughout this section that $G$ is loopless.
Our main tools are the edge-strength estimates underlying the Bencz\'ur--Karger cut sparsifier \cite{BK15} and Karger's bounds on the number of cuts \cite{Karger93,Karger99}.
We first recall the relevant definitions.

\begin{definition}[$k$-edge-connected]
  Let $G:=(V,E)$ be a graph.  For a subset $S\subseteq V$, let $\delta_G(S)$ denote the multiset of edges with exactly one endpoint in $S$; this multiset is called a \emph{cut} of $G$.  The cut is \emph{nontrivial} if $S$ is a nonempty proper subset of $V$.

  We say that $G$ is $k$-edge-connected if all nontrivial cuts of $G$ have size at least $k$, i.e., $|\delta_G(S)|\geq k$ for every nonempty proper subset $S\subset V$.
\end{definition}

In particular, a singleton vertex is considered $k$-edge-connected for any $k>0$.

\begin{definition}[Edge strength]\label{def:edge:strength}
  Fix a graph $G:=(V,E)$. For an edge $e=uv\in E$, its strength $s_e$ is the largest integer $k$ for which $u$ and $v$ are contained in a common induced subgraph that is $k$-edge-connected.
\end{definition}

In \cite{BK15}, Bencz\'ur and Karger introduced the notion of \emph{tight strength} $\widetilde s_e$ for every edge $e\in E$, which is a computable lower bound on $s_e$.
The following is the result for the unweighted case, which is sufficient for our purpose.

\begin{proposition}[\text{\cite[Theorem~6.5]{BK15}}]
  \label{prop:BK:tight:strength}
  Given an unweighted multigraph $G:=(V,E)$ with $n:=|V|$ and $m:=|E|$, one can compute in $O(m\log^2 n)$ time positive values $(\widetilde s_e)_{e\in E}$ such that $\widetilde s_e\leq s_e$ for every $e\in E$ and $\sum_{e\in E}1/\widetilde s_e\leq c_0n$, where $c_0>0$ is an absolute constant.
\end{proposition}

We also use the following cut-counting bound of Karger~\cite{Karger93,Karger99}.

\begin{proposition}[\text{\cite[Corollary~A.7]{Karger99}}]
  \label{prop:Karger:cut:counting}
Let $G:=(V,E)$ be a multigraph on $n\geq2$ vertices such that every nontrivial cut of $G$ has size at least $K$.  For every real $\alpha\geq1$, the number of distinct cuts of $G$ with size at most $\alpha K$ is at most $n^{2\alpha}$.
\end{proposition}

We are now ready to prove the Ising sparsification theorem.
\begin{proof}[Proof of \Cref{thm:Ising:sparsification}]
Apply \Cref{prop:BK:tight:strength} to $G$ to obtain $(\widetilde s_e)_{e\in E}$ satisfying
\begin{align}
  \widetilde s_e\leq s_e
  \quad\text{for every $e\in E$},
  \qquad
  \sum_{e\in E}\frac{1}{\widetilde s_e}\leq c_0n.
  \label{eqn:strength:estimates}
\end{align}
We then set
\begin{align}
  K:=\left\lceil
    \frac{\log(2n^5/\eta)}{\log((\beta+1)/2)}
  \right\rceil
  \label{eqn:sparsification:K}
\end{align}
and let $E_{\geq K}:=\{e\in E:\widetilde s_e\geq K\}$.  Let $C_1,\ldots,C_N$ be the connected components of $(V,E_{\geq K})$.  Construct $G'$ by contracting each component $C_i$ into a distinct single vertex, deleting all resulting self-loops, and retaining every edge between distinct components with its multiplicity.
Formally, $V':=\{C_1,\ldots,C_N\}$.  For any two distinct $C_i,C_j\in V'$, the number of parallel edges between $C_i$ and $C_j$ equals the number of edges in $E$ with endpoints in $C_i$ and $C_j$.

Since all edges in $E_{\geq K}$ are contracted, every edge $e$ of $G'$ satisfies $\widetilde s_e<K$ in $G$.  Hence
\begin{align*}
  |E'| = \sum_{e \in E'} 1
  \leq
  \sum_{e\in E'}\frac{K}{\widetilde s_e}
  \leq K \cdot c_0 n
  =O_\beta\bigl(n\log(n/\eta)\bigr),
\end{align*}
where the last inequality follows from \eqref{eqn:strength:estimates}, and $|V'|=N\leq n$.

We can lift an Ising configuration on $G'$ back to a configuration on $G$ by giving each contracted component the same spin as the corresponding vertex in $G'$.
We next bound the total variation distance between the Ising distribution on $G$ and the distribution lifted from $G'$.
For $U\subseteq V$, let $G[U]$ denote the subgraph of $G$ induced by $U$, and call $U$ a $K$-strong vertex set if $G[U]$ is $K$-edge-connected.  Let $B_1,\ldots,B_L$ be the inclusion-maximal $K$-strong vertex sets, and call the induced graphs $G[B_1],\ldots,G[B_L]$ the maximal $K$-strong components.

If two $K$-strong vertex sets $B_i$ and $B_j$ intersect, then their union is again $K$-strong.
To see this, any nonempty proper subset $S\subset B_i\cup B_j$ restricts to a nonempty proper subset of at least one of the two sets and hence the cut has size at least $K$.
Consequently, the maximal sets $B_1,\ldots,B_L$ are pairwise disjoint.  Every vertex belongs to one of these maximal sets because a singleton is $K$-strong.
Now consider an edge $e=uv\in E_{\geq K}$.  Since $\widetilde s_e\geq K$ and $\widetilde s_e\leq s_e$, the definition of $s_e$ implies that $u$ and $v$ belong to a common $K$-strong vertex set, and hence to the same maximal set $B_\ell$.  Following an $E_{\geq K}$-path shows that every connected component $C_i$ of $(V,E_{\geq K})$ is contained in a unique $B_\ell$.

Consider the random-cluster model with parameter $p:=1-1/\beta$ on $G$ through the Edwards--Sokal coupling in \Cref{prop:Edwards:Sokal}.
For any edge $e\in E$, the conditional probability that $e$ is open given the states of all other edges is at least
\begin{align*}
  x:=\min\left\{p,\frac{p}{p+2(1-p)}\right\}
   =\frac{p}{p+2(1-p)}
   =\frac{\beta-1}{\beta+1}.
\end{align*}
Thus, the random-cluster configuration stochastically dominates independent Bernoulli-$x$ bond percolation, where every edge is open independently with probability $x$.

Fix a maximal $K$-strong component $H=G[B_\ell]=(V_H,E_H)$ with at least two vertices, and let $n_H:=|V_H|\leq n$.  For $j\geq1$, let $\mathcal C_j$ be the collection of cuts of $H$ whose sizes belong to the interval $[jK,(j+1)K)$.
By \Cref{prop:Karger:cut:counting},
\begin{align*}
  |\mathcal C_j|
  \leq n_H^{2(j+1)}
  \leq n^{2(j+1)}.
\end{align*}
The percolated subgraph of $H$ can be disconnected only if all the edges of some nontrivial cut of $H$ are closed.
For a fixed cut of $H$ with size at least $jK$, under Bernoulli-$x$ bond percolation, the probability that all edges in the cut are closed is at most $(1-x)^{jK}$.
By a union bound,
\begin{align}
  \Pr(\text{the percolated subgraph of }H\text{ is disconnected}) &\leq \; \sum_{j\geq1}|\mathcal C_j|(1-x)^{jK}\notag \\
  &\leq \; \sum_{j\geq1}n^{2(j+1)}(1-x)^{jK}
  =\frac{n^4(1-x)^K}{1-n^2(1-x)^K}.
  \label{eqn:sparsification:cut:sum}
\end{align}
Moreover, $1-x=2/(\beta+1)$, and the choice of $K$ in \eqref{eqn:sparsification:K} implies $(1-x)^K =(2 / (\beta + 1))^K\leq \eta / (2n^5)$.
In particular, $n^2(1-x)^K\leq\eta/(2n^3)\leq1/2$.  Substituting this estimate into \eqref{eqn:sparsification:cut:sum} yields
\begin{align*}
  \Pr(\text{the percolated subgraph of }H\text{ is disconnected})
  &\leq \frac{n^4(1-x)^K}{1-n^2(1-x)^K}
  \leq 2n^4(1-x)^K
  \leq \frac{\eta}{n},
\end{align*}
where singleton components are trivially connected.
Since connectivity is an increasing event, stochastic domination gives the same upper bound under the random-cluster measure.  There are at most $n$ maximal $K$-strong components, so another union bound shows that all of them are connected in the random-cluster configuration with probability at least $1-\eta$.  Under the Edwards--Sokal coupling, each open component receives the same spin.  Therefore, under the Ising model on graph $G$, it holds that
\begin{align*}
  \mu^{\mathrm{Ising}}_{\beta,G}(\mathcal E)\geq1-\eta,
\end{align*}
where $\mathcal E$ is the event that all vertices in each block $C_i$ are assigned the same spin.

Store for each $v\in V$ the index of the block $C_i$ containing it.  This is an implicit representation of the lifting map: for $\sigma'\in\{0,1\}^{V'}$, set $\Lift(\sigma')_v:=\sigma'_{C_i}$ whenever $v\in C_i$.

Let $m_{\mathrm{in}}:=m-|E'|$.  For every $\sigma'\in\{0,1\}^{V'}$, $\mono(\Lift(\sigma')) =m_{\mathrm{in}}+\mono_{G'}(\sigma')$.
The first term is independent of $\sigma'$, and hence $\Lift(\mu')=\mu^{\mathrm{Ising}}_{\beta,G}(\,\cdot\mid\mathcal E)$.
Consequently,
\begin{align*}
  \TV(\Lift(\mu'),\mu^{\mathrm{Ising}}_{\beta,G})
  =1-\mu^{\mathrm{Ising}}_{\beta,G}(\mathcal E)
  \leq\eta.
\end{align*}

The components of $(V,E_{\geq K})$, the graph $G'$, and the block-index representation of $\Lift$ require only $O(m+n)$ additional time after the edge strength estimation.
This proves the theorem.
\end{proof}

\section{Preliminaries for Markov chain analysis}
\label{sec:Markov:preliminaries}

We begin with the basic notions of finite Markov chains.
Let $\Omega$ be a finite state space.  A Markov chain $(X_t)_{t\geq0}$ on $\Omega$ is specified by a transition matrix $P$, where $P(x,y)$ is the probability of moving from $x$ to $y$ in one step. Let $P^t(x,y)$ denote the probability of moving from $x$ to $y$ in $t$ steps and write $P^t(x,\cdot)$ for the distribution of $X_t$ conditional on $X_0=x$.
A probability distribution $\pi$ on $\Omega$ is stationary if $\pi P=\pi$.
The chain is irreducible if, for every $x,y\in\Omega$, there exists $t\geq0$ such that $P^t(x,y)>0$.  The period of $x$ is the greatest common divisor of the set $\{t\geq1:P^t(x,x)>0\}$, and the chain is aperiodic if every state has period one.
Every finite irreducible and aperiodic Markov chain has a unique stationary distribution $\pi$, and $P^t(x,\cdot)$ converges to $\pi$ from every initial state $x\in\Omega$.
The chain is reversible with respect to $\pi$ if, for every pair $x,y\in\Omega$, the detailed balance condition $\pi(x)P(x,y)=\pi(y)P(y,x)$ holds.
Moreover, reversibility implies that $\pi$ is stationary.

\subsection{Covariance and the Dirichlet Form}

For functions $f,g:\Omega\to\mathbb R$, let $X\sim\pi$, write $F:=f(X)$ and $G:=g(X)$, and define $\langle f,g\rangle_\pi:=\mathbb E[FG]$.
Throughout, $\operatorname{Cov}(U,V):=\mathbb E[UV]-\mathbb E[U]\mathbb E[V]$ is taken under the probability law specified by the surrounding context, and $\operatorname{Cov}(U,V\mid Z)$ denotes conditional covariance given $Z$.
We recall the law of total covariance:
\begin{equation}\label{eqn:total:covariance}
  \operatorname{Cov}(U,V)=\mathbb E[\operatorname{Cov}(U,V\mid Z)]+\operatorname{Cov}(\mathbb E[U\mid Z],\mathbb E[V\mid Z]).
\end{equation}
Let $(U',V')$ have the same joint distribution as $(U,V)$ and be independent of $(U,V)$; we call $(U',V')$ an independent copy of $(U,V)$.
The two-replica identity is
\begin{equation}\label{eqn:two:replica:identity}
  \operatorname{Cov}(U,V)=\mathbb E[(U-U')(V-V')] / 2.
\end{equation}
Indeed, by independence, $\mathbb E[UV']=\mathbb E[U'V]=\mathbb E[U]\mathbb E[V]$, while $\mathbb E[U'V']=\mathbb E[UV]$.

Suppose that $P$ is reversible with respect to $\pi$, and regard $P$ as the operator $(Pg)(x):=\sum_{y\in\Omega}P(x,y)g(y)$.
Detailed balance implies that $P$ is self-adjoint with respect to $\langle\cdot,\cdot\rangle_\pi$, namely $\langle Pf,g\rangle_\pi=\langle f,Pg\rangle_\pi$.
The Dirichlet form associated with $P$ is defined by
\begin{equation}\label{eqn:Dirichlet:form}
  \mathcal E_P(f,g):=\frac{1}{2}\sum_{x,y\in\Omega}\pi(x)P(x,y)\bigl(f(x)-f(y)\bigr)\bigl(g(x)-g(y)\bigr)=\langle f,(I-P)g\rangle_\pi.
\end{equation}

We record the following standard identity for random-scan heat-bath dynamics.
Suppose that $\pi$ is a probability distribution on $2^E$, let $m:=|E|$, and let $P$ be its random-scan single-edge heat-bath dynamics.
The next lemma states that the Dirichlet form is the average of local changes caused by updating a single edge. Let $X \sim \pi$. {For each $e\in E$, write $X_{-e}:=X\setminus\{e\}$ for the configuration on the remaining edges.} For any functions $f,g:2^E\to\mathbb R$, define two random variables $F:=f(X)$ and $G:=g(X)$.

\begin{lemma}\label{lem:heat:bath:Dirichlet}
For all functions $f,g:2^E\to\mathbb R$,
\begin{equation}\label{eqn:heat:bath:Dirichlet}
  \mathcal E_P(f,g)=\frac{1}m\sum_{e\in E}\mathbb E\bigl[\operatorname{Cov}(F,G\mid X_{-e})\bigr].
\end{equation}
\end{lemma}

\begin{proof}
  For each $e\in E$, let $P_e$ be the heat-bath kernel that resamples only $e$. Then $P = \frac{1}m\sum_{e\in E}P_e$.
  View any set $A \subseteq E$ as a vector in $\{0,1\}^E$ by setting $A_e = 1$ if $e \in A$ and $A_e = 0$ otherwise.
  Then $\pi$ is a distribution over $\{0,1\}^E$. View $f,g$ as functions from $\{0,1\}^E$ to $\mathbb R$. 
  {Write $E-e:=E\setminus\{e\}$. For $x\in\{0,1\}^{E-e}$ and $c\in\{0,1\}$, let
  \begin{equation*}
    \pi_{E-e}(x):=\Pr(X_{-e}=x),\qquad
    \pi_e^x(c):=\Pr(X_e=c\mid X_{-e}=x)
  \end{equation*}
  denote the marginal and conditional probabilities, respectively. When $\pi_{E-e}(x)=0$, the conditional law $\pi_e^x$ may be chosen arbitrarily.}
  By the definition of the Dirichlet form, we have
\begin{equation*}
  \mathcal E_P(f,g)=\frac{1}{m}\sum_{e \in E}\sum_{x \in \{0,1\}^{E -e}}\pi_{E-e}(x) \cdot \frac{1}{2}\sum_{c_0,c_1 \in \{0,1\}}\pi^x_e(c_0)\pi^x_e(c_1)(f(x_0) - f(x_1))(g(x_0) - g(x_1)),
\end{equation*}
where $x_i$ is the vector in $\{0,1\}^E$ that is equal to $x$ on $E-e$ and takes the value $c_i$ on $e$. By~\eqref{eqn:two:replica:identity}, the above equals $\frac{1}{m}\sum_{e \in E}\sum_{x \in \{0,1\}^{E-e}}\pi_{E-e}(x) [\operatorname{Cov}(F,G\mid X_{-e}=x)] $, as desired.
\end{proof}

If a distribution $\nu$ is absolutely continuous with respect to $\pi$, write $h_\nu(x):=\nu(x)/\pi(x)$ for its density.
Then $\mathbb E_\nu[f]=\langle h_\nu,f\rangle_\pi$, $\mathbb E_\pi[h_\nu]=1$, and $\operatorname{Cov}(h_\nu(X),f(X))=\mathbb E_\nu[f]-\mathbb E_\pi[f]$.
If $\mu_t:=\mu P^t$ and $h_t:=h_{\mu_t}$, reversibility implies $h_{t+1}=Ph_t$.
When $\mu$ is the point mass at $z\in\Omega$, detailed balance of $P$ further shows that $h_t(x)=P^t(z,x)/\pi(x)=P^t(x,z)/\pi(z)$.

\subsection{Monotone Coupling}

Let $(\Omega,\preceq)$ be a finite partially ordered state space.
A function $f:\Omega\to\mathbb R$ is increasing if $x\preceq y$ implies $f(x)\leq f(y)$, and it is decreasing if $-f$ is increasing.

\begin{definition}\label{def:monotone:grand:coupling}
A grand coupling of a Markov kernel $P$ is a collection $(X_t^x)_{x\in\Omega}$ of chains defined on a common probability space such that, for every $x\in\Omega$, $(X_t^x)_{t\geq0}$ has transition kernel $P$ and starts from $X_0^x=x$.
The grand coupling is monotone if $x\preceq y$ implies $X_t^x\preceq X_t^y$ for every $t\geq0$.
\end{definition}

For any $x\in\Omega$ and a test function $f:\Omega\to\mathbb R$, $(P^t f)(x)=\mathbb E[f(X_t^x)]$.
Thus, a monotone grand coupling implies that $P^t f$ is increasing whenever $f$ is increasing.
For the random-cluster model, $\Omega=2^E$, the same coupling also bounds the distance to stationarity by the gap between the all-closed and all-open chains, as stated next.

\begin{lemma}\label{lem:extremal:coupling}
Let $\Omega=2^E$ be ordered by set inclusion, with $\mathbf 0:=\emptyset$ and $\mathbf 1:=E$. Let $P$ be a Markov kernel with stationary distribution $\pi$, and fix a monotone grand coupling $(X_t^x)_{x\in\Omega}$ of $P$.
For $Y_0\sim\pi$ independent of the randomness used in the grand coupling, set $Y_t:=X_t^{Y_0}$, $X_t^-:=X_t^{\mathbf 0}$, and $X_t^+:=X_t^{\mathbf 1}$.
Then, for every $A\subseteq E$ and $t\geq0$,
\begin{equation}\label{eqn:extremal:coupling}
  \TV(P^t(A,\cdot),\pi)\leq\Pr(X_t^A\neq Y_t)\leq\Pr(X_t^-\neq X_t^+)\leq\mathbb E\bigl[|X_t^+\setminus X_t^-|\bigr].
\end{equation}
\end{lemma}

\begin{proof}
Since $Y_t\sim\pi$, the first inequality is the coupling inequality.
The inclusions $X_t^-\subseteq X_t^A\subseteq X_t^+$ and $X_t^-\subseteq Y_t\subseteq X_t^+$ prove the second inequality.
Finally, taking expectations in $\mathbf 1\{X_t^-\neq X_t^+\}\leq|X_t^+\setminus X_t^-|$ proves the last inequality.
\end{proof}

\section{Mixing time of random-cluster Glauber dynamics}

After using the sparsification in \Cref{thm:Ising:sparsification}, we get a multigraph $G$.
Any loop corresponds to an independent random variable in the Gibbs distribution, and we can sample it separately.
Thus, throughout this section, we assume that $G=(V,E)$ is connected and loopless.
Let $n:=|V|$, $m:=|E|$, $\pi:=\mu^{\mathrm{RC}}_{p,G}$, and $\rho:=p/(2-p)$.

We first use a monotone grand coupling to reduce the mixing-time bound to exponential decay of the two extremal open-edge biases.
We then state a monotone edge-count Poincar\'e inequality, use it to prove this decay, and finally prove the inequality.

\subsection{Mixing from the Extremal Coupling}

We begin by placing the chains from all initial configurations on the same probability space.
The order-preserving update then allows every chain to be compared with the two extremal chains.

\begin{lemma}\label{lem:RC:monotone:grand:coupling}
The random-cluster Glauber dynamics admits a monotone grand coupling $\{X_t^A\}$, indexed by $A\subseteq E$ and $t\geq0$, with $X_0^A=A$.
\end{lemma}

\begin{proof}
Consider the following implementation of the transition of Glauber dynamics. Given the current set $S$ of open edges, sample $e \in E$ uniformly at random. Calculate the conditional probability of opening $e$ given $S\setminus \{e\}$. 
This would be either $p$ or $\rho$, depending on whether the endpoints of $e$ are connected in $S\setminus \{e\}$.
Call the threshold $p_0$.
Then, sample a random number $r \in [0,1]$ uniformly at random. If $r \leq p_0$, then set $e$ to open. Otherwise, close $e$.
Couple all chains by using the same uniformly chosen edge and the same independent uniform random number at every step.
If $A\subseteq B$ and the chosen edge is $e=uv$, then $A\setminus\{e\}\subseteq B\setminus\{e\}$, so connectivity of $u$ and $v$ in the first configuration implies connectivity in the second.
The conditional probability of opening $e$ is therefore no larger for $A$ than for $B$, because it is $\rho$ when the endpoints are disconnected and $p$ when they are connected, with $\rho\leq p$.
Thus, the common update preserves inclusion, and induction shows that $X_t^A\subseteq X_t^B$ whenever $A\subseteq B$.
\end{proof}

Let $\mathbf 0:=\emptyset$ and $\mathbf 1:=E$, and write $X_t^-:=X_t^{\mathbf 0}$ and $X_t^+:=X_t^{\mathbf 1}$ for the all-closed and all-open chains in this coupling.
Define their lower and upper open-edge biases by
\begin{equation}\label{eqn:RC:extremal:biases}
  \Delta_t^-:=\mathbb E_{X\sim\pi}[|X|]-\mathbb E[|X_t^-|],\qquad \Delta_t^+:=\mathbb E[|X_t^+|]-\mathbb E_{X\sim\pi}[|X|].
\end{equation}
Under a coupling of two extremal chains $(X_t^-)_{t\geq0},(X_t^+)_{t\geq0}$ with a stationary chain $(Y_t)_{t\geq0}$ with $Y_0 \sim \pi$, $X_t^-\subseteq Y_t\subseteq X_t^+$, and hence $\Delta_t^-,\Delta_t^+\geq0$.
The main ingredient in the proof of \Cref{thm:RC:Glauber:mixing} is to show that these biases decay exponentially in $t$ with a rate of $1/(2nm)$, as stated in the following lemma.

\begin{lemma}\label{lem:RC:extremal:bias:decay}
For every integer $t\geq0$,
\begin{equation}\label{eqn:RC:extremal:bias:decay}
  \Delta_t^-\leq m\exp\left(-\frac{t}{2nm}\right)\qquad\text{and}\qquad \Delta_t^+\leq m\exp\left(-\frac{t}{2nm}\right).
\end{equation}
\end{lemma}

We first prove \Cref{thm:RC:Glauber:mixing} assuming \Cref{lem:RC:extremal:bias:decay} and then prove the lemma.
\begin{proof}[Proof of \Cref{thm:RC:Glauber:mixing}]
Let $A\subseteq E$ be arbitrary.
By \Cref{lem:extremal:coupling}, \eqref{eqn:RC:extremal:biases}, and \Cref{lem:RC:extremal:bias:decay},
\begin{equation*}
  \TV(P^t(A,\cdot),\pi)\leq\mathbb E\bigl[|X_t^+\setminus X_t^-|\bigr]=\Delta_t^++\Delta_t^-\leq2m\exp\left(-\frac{t}{2nm}\right).
\end{equation*}
Choosing $t=\left\lceil2nm\log(2m/\eps)\right\rceil$ makes this bound at most $\eps$ uniformly over all $A\subseteq E$, which proves the claimed mixing-time bound by the definition of $\Tmix(P,\eps)$.
\end{proof}

It remains to prove \Cref{lem:RC:extremal:bias:decay}.
The central ingredient is a bilinear Poincar\'e inequality between monotone functions and the edge-count observable, or a monotone edge-count Poincar\'e inequality for short.
It controls a global covariance with $|X|$ by single-edge conditional covariances.
The left-hand side measures the correlation between a monotone function $f(X)$ and the total number of open edges, whereas each term on the right isolates the correlation contributed by one edge after all other edges are fixed.

\begin{lemma}[Monotone edge-count Poincar\'e inequality]\label{lem:RC:monotone:Poincare}
Let $X\sim\pi$.
For every increasing function $f:2^E\to\mathbb R$, let $F:=f(X)$. Then,
\begin{equation}\label{eqn:RC:monotone:Poincare}
  \operatorname{Cov}(F,|X|)\leq2n\sum_{e\in E}\mathbb E_\pi\bigl[\operatorname{Cov}(F,|X|\mid X_{-e})\bigr].
\end{equation}
\end{lemma}

By \Cref{lem:heat:bath:Dirichlet}, we can rewrite the right-hand side of \eqref{eqn:RC:monotone:Poincare} as $2nm\cdot\mathcal E_P(f,|\cdot|)$, where $\mathcal E_P$ is the Dirichlet form of the Glauber dynamics.
This makes it similar to a standard Poincar\'e inequality, except that instead of a single function $f$ on both coordinates, we have a monotone function $f$ and the edge-count observable $|\cdot|$.
We note that one may also hope for a bilinear Poincar\'e inequality with both $f$ and $g$ monotone, but this is false in general.
It is important that one of the two functions is the edge-count observable.

Next we prove \Cref{lem:RC:extremal:bias:decay} assuming the monotone edge-count Poincar\'e inequality.

\begin{proof}[Proof of \Cref{lem:RC:extremal:bias:decay}]
Let $X\sim\pi$.
For $A\subseteq E$, define $h_t^+(A):=P^t(\mathbf 1,A)/\pi(A)$ and $h_t^-(A):=P^t(\mathbf 0,A)/\pi(A)$, namely, the densities of $X_t^+$ and $X_t^-$ with respect to $\pi$.
For the all-open chain, the detailed balance condition implies $h_t^+(A)=P^t(A,\mathbf 1)/\pi(\mathbf 1)$.
From this, we can see that $h_t^+$ is increasing in $A$ under set inclusion. 
If $A\subseteq B$, then by the monotone grand coupling, when the chain from $A$ reaches the all-open configuration, the chain from $B$ must also be at the all-open configuration.
Thus, $P^t(A,\mathbf 1)\leq P^t(B,\mathbf 1)$ and monotonicity follows.

Since $\mathbb E_\pi[h_t^+]=1$, we have $\Delta_t^+=\operatorname{Cov}(h_t^+(X),|X|)$.
Therefore, combining \Cref{lem:RC:monotone:Poincare} with \Cref{lem:heat:bath:Dirichlet} implies that
\begin{equation*}
  \Delta_t^+\leq2n\sum_{e\in E}\mathbb E\bigl[\operatorname{Cov}(h_t^+(X),|X|\mid X_{-e})\bigr]=2nm \cdot \mathcal E_P(h_t^+,|\cdot|).
\end{equation*}
By self-adjointness of $P$ and the identity $h_{t+1}^+=Ph_t^+$,
\begin{equation*}
  \mathcal E_P(h_t^+,|\cdot|)=\langle(I-P)h_t^+,|\cdot|\rangle_\pi=\langle h_t^+-h_{t+1}^+,|\cdot|\rangle_\pi=\Delta_t^+-\Delta_{t+1}^+.
\end{equation*}
Consequently, $\Delta_{t+1}^+\leq(1-1/(2nm))\Delta_t^+$.
Since $\Delta_0^+=m-\mathbb E_\pi[|X|]\leq m$, iteration shows that
\begin{equation*}
  \Delta_t^+\leq m\left(1-\frac{1}{2nm}\right)^t\leq m\exp\left(-\frac{t}{2nm}\right).
\end{equation*}

For the all-closed chain, detailed balance similarly shows that $h_t^-$ is decreasing, so $-h_t^-$ is increasing.
Moreover, $\Delta_t^-=\operatorname{Cov}(-h_t^-(X),|X|)$ and $\mathcal E_P(-h_t^-,|\cdot|)=\Delta_t^--\Delta_{t+1}^-$.
Applying \Cref{lem:RC:monotone:Poincare} to $-h_t^-$ and using this Dirichlet identity shows that $\Delta_{t+1}^-\leq(1-1/(2nm))\Delta_t^-$; since $\Delta_0^-=\mathbb E_\pi[|X|]\leq m$, iteration proves the lower-bias bound.
\end{proof}

\subsection{Even-Subgraph Model}

To prove the monotone edge-count Poincar\'e inequality, we first introduce the even-subgraph model (also known as the subgraph world \cite{JS93} or high-temperature expansion).
Recall that $\pi:=\mu^{\mathrm{RC}}_{p,G}$ and $\rho:=p/(2-p)$.
For $S\subseteq E$, let $\partial S$ be the set of vertices of odd degree in $(V,S)$, where parallel edges are counted with multiplicity.
For $B\subseteq V$, define
\begin{equation*}
  \Omega_B:=\{S\subseteq E:\partial S=B\},\qquad Z_B:=\sum_{S\in\Omega_B}\rho^{|S|},
\end{equation*}
and abbreviate $Z_0:=Z_\emptyset$.
The even-subgraph distribution is supported on $\Omega_\emptyset$ and defined by
\begin{equation*}
  \forall S\in\Omega_\emptyset, \quad \nu_0(S):=\frac{\rho^{|S|}}{Z_0}.
\end{equation*}
The even-subgraph model and the random-cluster model are related via the following coupling by Grimmett and Janson \cite{GJ09}.

\begin{proposition}
  \label{prop:Grimmett:Janson:coupling}
  Sample $S\sim\nu_0$.
  Form $R\subseteq E$ by including every edge of $S$ and including each edge of $E\setminus S$ independently with probability $\rho$.
  Then $R\sim\pi$.
\end{proposition}

For $S,R\subseteq E$, the lifting kernel in \Cref{prop:Grimmett:Janson:coupling} is
\begin{equation}\label{eqn:RC:lift:kernel}
  K(S,R):=\mathbf 1\{S\subseteq R\} \cdot \rho^{|R\setminus S|}(1-\rho)^{m-|R|}.
\end{equation}
The following identity extends this coupling to the families $\Omega_B$ defined above.

\begin{lemma}
\label{lem:RC:ES:lift}
For every $B\subseteq V$ and $R\subseteq E$,
\begin{equation}\label{eqn:RC:ES:lift}
  \frac{1}{Z_0}\sum_{S\in\Omega_B}\rho^{|S|}K(S,R)
  =\mathbf 1\left\{\begin{gathered}
    \text{every component of $(V,R)$}\\
    \text{contains an even number of vertices of $B$}
  \end{gathered}\right\} \cdot \pi(R).
\end{equation}
\end{lemma}

\begin{proof}
Only sets $S\subseteq R$ contribute to the left-hand side of \eqref{eqn:RC:ES:lift}, and for every such $S$,
\begin{equation*}
  \rho^{|S|}K(S,R)=\rho^{|R|}(1-\rho)^{m-|R|},
\end{equation*}
which is independent of $S$.
Fix a component $C$ of $(V,R)$.
Since $S\subseteq R$, every edge of $S$ incident to a vertex of $V(C)$ lies in $E(C)$, so each $v\in V(C)$ has the same degree in $(V,S)$ as in the subgraph $(V(C),S\cap E(C))$.
Hence, $V(C)\cap\partial S$ is exactly the set of odd-degree vertices in $(V(C),S\cap E(C))$, and the handshaking lemma shows that $|V(C)\cap\partial S|$ is even.
If $\partial S=B$, then $V(C)\cap\partial S=V(C)\cap B$, so $|V(C)\cap B|$ must also be even.
Consequently, if some component of $(V,R)$ contains an odd number of vertices of $B$, no set $S\subseteq R$ can satisfy $\partial S=B$, and the left-hand side is zero.

Therefore the lemma follows when $R$ does not satisfy the parity condition.
Now fix $R\subseteq E$ such that every component of $(V,R)$ contains an even number of vertices of $B$.
We first show that the number of sets $S\subseteq R$ satisfying $\partial S=B$ is $2^{|R|-n+\kappa(R)}$, where $\kappa(R)$ is the number of components of $(V,R)$; in particular, this number does not depend on $B$.
This follows easily by considering the linear system over $\mathbb F_2$ given by the incidence matrix of $(V,R)$, where each edge variable is $1$ if the corresponding edge is in $S$ and $0$ otherwise.
The right-hand side of the system is the indicator vector of $B$.
The rank of this system is $n-\kappa(R)$.
The parity requirement guarantees that a solution exists, and the number of solutions is $2^{|R|-n+\kappa(R)}$.
See also \cite[Lemma 12]{GJ18}.

Since the same count holds for $B=\emptyset$ and every contributing summand is the same (recall that $\rho^{|S|}K(S,R)$ is independent of $S$),
\begin{align*}
  \frac{1}{Z_0}\sum_{S\in\Omega_B}\rho^{|S|}K(S,R)
  &=\frac{2^{|R|-n+\kappa(R)}\rho^{|R|}(1-\rho)^{m-|R|}}{Z_0}
  =\frac{1}{Z_0}\sum_{S\in\Omega_\emptyset}\rho^{|S|}K(S,R)\\
  &=\sum_{S\in\Omega_\emptyset}\nu_0(S)K(S,R)=\pi(R). \qedhere
\end{align*}

\end{proof}

The following is a direct corollary of \Cref{lem:RC:ES:lift}.
\begin{corollary}\label{cor:RC:sector:partition:bound}
For every $B\subseteq V$,
\begin{equation}\label{eqn:RC:sector:partition:bound}
  Z_B=Z_0\Pr_{R\sim\pi}\left(
    \begin{gathered}
      \text{every component of $(V,R)$ contains an even}\\
      \text{number of vertices of $B$}
    \end{gathered}
  \right)\leq Z_0.
\end{equation}
\end{corollary}

\begin{proof}
Summing the left-hand side of \eqref{eqn:RC:ES:lift} over $R\subseteq E$ gives
\begin{equation*}
  \frac{1}{Z_0}\sum_{S\in\Omega_B}\rho^{|S|}\sum_{R\subseteq E}K(S,R)=\frac{1}{Z_0}\sum_{S\in\Omega_B}\rho^{|S|}=\frac{Z_B}{Z_0}.
\end{equation*}
Here we used $\sum_{R\subseteq E}K(S,R)=1$.
Summing the right-hand side of \eqref{eqn:RC:ES:lift} gives the probability in \eqref{eqn:RC:sector:partition:bound}.
Multiplying by $Z_0$ proves the identity.
\end{proof}

The partition-function bound above will control the sums over $\Omega_B$ that arise when an even subgraph is changed one edge at a time.
Define
\begin{align*}
  \Omega_2:=\{S\subseteq E:|\partial S|=2\},\qquad
  \Omega_{\leq 2}:=\Omega_\emptyset\cup\Omega_2.
\end{align*}
Configurations in $\Omega_2$ are called near-even subgraphs, and $\Omega_{\leq 2}$ is the state space of the so-called worm process \cite{PS01}.
For edge sets $S,T\subseteq E$, let $S\oplus T$ be their symmetric difference.

Since $G$ is loopless, if $S\in\Omega_\emptyset$ and $e=uv$, then flipping $e$ changes the parity of the degrees of $u$ and $v$, so $\partial(S\oplus\{e\})=\{u,v\}$ and $S\oplus\{e\}\in\Omega_2$.
 In \cite{GJ18}, the worm state space is used to construct canonical paths, but here we only use it to express the covariance as a sum of local differences.

\begin{lemma}\label{lem:RC:worm:covariance}
Let $I\sim\nu_0$.
For every function $f:\Omega_{\leq 2}\to\mathbb R$,
\begin{equation}\label{eqn:RC:worm:covariance}
  \left|\operatorname{Cov}(f(I),|I|)\right|
  \leq\frac{1}{Z_0}
  \sum_{\substack{S\in\Omega_{\leq 2},\,e\in E\setminus S\\
      S\cup\{e\}\in\Omega_{\leq 2}}}
  \rho^{|S|+1}\left|f(S\cup\{e\})-f(S)\right|.
\end{equation}
\end{lemma}

The proof routes the symmetric difference of two independent even subgraphs along edge-disjoint cycles, keeping every intermediate configuration in $\Omega_{\leq 2}$.
An injective encoding and \Cref{cor:RC:sector:partition:bound} then control the total weight assigned to each routed edge flip.

\begin{proof}
Let $J\sim\nu_0$ be an independent copy of $I$. 
Then, by \eqref{eqn:two:replica:identity},
\begin{align}
  \operatorname{Cov}(f(I),|I|) & = \frac{1}{2}\cdot\Ex[(f(I)-f(J))(|I|-\abs{J})]. \nonumber
\end{align}
Since $\nu_0(I)=\rho^{\abs{I}}/Z_0$,
we have
\begin{equation}\label{eqn:RC:two:replica:cardinality}
  T \defeq 2Z_0^2 \cdot \left|\operatorname{Cov}(f(I),|I|)\right|
  = \left|\sum_{(I,J)\in\Omega_\emptyset^2}\rho^{|I|+|J|}\bigl(f(I)-f(J)\bigr)\bigl(|I|-|J|\bigr)\right|.
\end{equation}
For an ordered pair $(I,J)$, let $D:=I\oplus J$.
For $e\in D$, set a signing $\sigma$ so that $\sigma_e:=1$ if $e\in I$ and $\sigma_e:=-1$ if $e\in J$. Note that $\sigma$ is fixed by $I$ and $J$.
Thus, $|I|-|J|=\sum_{e\in D}\sigma_e$.
Since $I$ and $J$ are even subgraphs, every vertex has even degree in $D$.
Thus, $D$ can be decomposed into edge-disjoint simple cycles.
Fix one such decomposition as a function of $D$ and orient each cycle.
Here we allow two parallel edges to form a cycle of length two.
For $e\in D$, let $Q_e$ be the oriented cycle containing $e$.

Fix $D$ and an edge $e\in D$.
Throughout the following calculation, all sums are over ordered pairs $(I,J)\in\Omega_\emptyset^2$ satisfying $I\oplus J=D$.
We claim that
\begin{align}\label{eqn:cycle:flip}
  \sum_{\substack{(I,J)\in\Omega_\emptyset^2\\I\oplus J=D}} \rho^{|I|+|J|}\sigma_e\cdot f(J)=-\sum_{\substack{(I,J)\in\Omega_\emptyset^2\\I\oplus J=D}} \rho^{|I|+|J|}\sigma_e\cdot f(I)=\sum_{\substack{(I,J)\in\Omega_\emptyset^2\\I\oplus J=D}} \rho^{|I|+|J|}\sigma_e\cdot f(I\oplus Q_e).
\end{align}
The first equality follows from swapping $I$ and $J$,
which preserves the summation range and the weight $\rho^{|I|+|J|}$ but reverses the sign of $\sigma_e$.
The second equality is obtained via the bijective map $(I,J)\mapsto(I\oplus Q_e,J\oplus Q_e)$, which preserves the summation range because flipping a cycle keeps both configurations even and $(I\oplus Q_e)\oplus(J\oplus Q_e)=D$.
Moreover, it merely moves the edges of $Q_e$ from one configuration to the other, so $|I|+|J|$ and the weight are unchanged.
Since $e\in Q_e$, it also reverses the sign of $\sigma_e$.

By \eqref{eqn:cycle:flip}, the sum containing $f(J)$ can be replaced by the sum containing $f(I\oplus Q_e)$.
Consequently, we have
\begin{equation}\label{eqn:RC:cycle:switching}
  \sum_{\substack{(I,J)\in\Omega_\emptyset^2\\I\oplus J=D}}
  \rho^{|I|+|J|}\sigma_e\cdot\bigl(f(I)-f(J)\bigr) =
  \sum_{\substack{(I,J)\in\Omega_\emptyset^2\\I\oplus J=D}}
  \rho^{|I|+|J|}\sigma_e\cdot
  \bigl(f(I)-f(I\oplus Q_e)\bigr).
\end{equation}

Consider the following process.
From $I$, flip only the edges of $Q_e$ one at a time in its orientation, starting from $e$, until reaching $I\oplus Q_e$.
Before any edge is flipped and after the entire cycle is flipped, the current configuration is even.
At every intermediate step, the flipped edges form a path, so the current configuration has exactly two odd-degree vertices.
Thus, every configuration encountered belongs to $\Omega_{\leq 2}$, and $f(I)-f(I\oplus Q_e)$ can be written as a telescoping sum of single-edge differences.

Consider one such single-edge flip.
Let $S$ be the current configuration.
Let $U$ be obtained from $J$ by flipping exactly the edges in $I\oplus S$.
In other words, $U=J\oplus(I\oplus S)=S\oplus D$.
Let $a$ be the next edge to be flipped. 
The triple $(S,U,a)$ uniquely determines the original pair $(I,J)$ and the starting edge $e$.
Indeed, given such a triple, we first recover $D=S\oplus U$.
Then, the fixed oriented cycle $Q_a$ is identified, which contains the starting edge $e$.
If $\partial S=\emptyset$, no edge of $Q_a$ has yet been flipped and the initial edge is $e=a$.
Otherwise, the tail of $a$ is one vertex of $\partial S$, while the other vertex in $\partial S$ is the tail of $e$.
Hence $e$ is the unique edge of $Q_a$ directed out of the other vertex.
The edges encountered from $e$ up to, but not including, $a$ are exactly those already flipped on $Q_a$.
Reversing these flips in $S$ and $U$ recovers $I$ and $J$.
Note that this step is very similar to the canonical path construction (and encoding) of Jerrum and Sinclair \cite{JS93}, except that we route the cycle difference along a fixed orientation rather than an arbitrary one.

Since $I$ and $J$ are even subgraphs and the same edges are flipped in both to obtain $S$ and $U$, $S$ and $U$ have the same odd-degree vertices.
Namely, $\partial S=\partial U$.
Moreover, the edge $a$ belongs to exactly one of $S$ and $U$, and $\rho^{|I|+|J|}=\rho^{|S|+|U|}$.
Let $\mathcal R$ be the set of all triples $(S,U,a)$ arising from the routed cycle differences above.
Using \eqref{eqn:RC:two:replica:cardinality}, we bound $T=2Z_0^2\cdot|\operatorname{Cov}(f(I),|I|)|$ by
\begin{equation*}
  \begin{aligned}
    T& = \left|\sum_{(I,J)\in\Omega_\emptyset^2}\rho^{|I|+|J|}\bigl(f(I)-f(J)\bigr)\bigl(|I|-|J|\bigr)\right| =\left|\sum_{D\in\Omega_\emptyset}\sum_{e\in D}\sum_{\substack{(I,J)\in\Omega_\emptyset^2\\I\oplus J=D}}\rho^{|I|+|J|}\sigma_e\cdot\bigl(f(I)-f(J)\bigr)\right|\\
    & =\left|\sum_{D\in\Omega_\emptyset}\sum_{e\in D}\sum_{\substack{(I,J)\in\Omega_\emptyset^2\\I\oplus J=D}}\rho^{|I|+|J|}\sigma_e\cdot\bigl(f(I)-f(I\oplus Q_e)\bigr)\right| \leq\sum_{(S,U,a)\in\mathcal R}\rho^{|S|+|U|}\left|f(S\oplus\{a\})-f(S)\right|,
  \end{aligned}
\end{equation*}
where the second equality uses $|I|-|J|=\sum_{e\in D}\sigma_e$, while the third follows from \eqref{eqn:RC:cycle:switching}.
The final inequality is obtained by routing each cycle difference into single-edge flips and applying the triangle inequality.

Grouping the triples in $\mathcal R$ by the pair $(S,a)$, we obtain
\begin{align*}
  \sum_{(S,U,a)\in\mathcal R}\rho^{|S|+|U|}\left|f(S\oplus\{a\})-f(S)\right|
  = \sum_{\substack{S\in\Omega_{\leq 2},\,a\in E\\S\oplus\{a\}\in\Omega_{\leq 2}}}\rho^{|S|}\left|f(S\oplus\{a\})-f(S)\right|
    \sum_{\substack{U\subseteq E\\(S,U,a)\in\mathcal R}}\rho^{|U|}.
\end{align*}
To bound the last sum, let $a=uv$, $B=\partial S$, and $B':=\partial(S\oplus\{a\})=B\oplus\{u,v\}$.
Every possible $U$ satisfies $\partial U=B$, and $a$ belongs to exactly one of $S$ and $U$.
Since all weights are nonnegative, we may enlarge the sum to all $U\in\Omega_B$ satisfying this condition on $a$.
We bound the $U$-sum in the two cases below using \eqref{eqn:RC:sector:partition:bound} in \Cref{cor:RC:sector:partition:bound}.
\begin{itemize}
  \item If $a\notin S$, then $a\in U$.
  Deleting $a=uv$ from $U$ decreases the degrees of $u$ and $v$ by one and leaves all other degrees unchanged. Hence $\partial(U\setminus\{a\})=B'$. Thus $U'=U\setminus\{a\}$ is a bijection from $\{U\in\Omega_B:a\in U\}$ to $\{U'\in\Omega_{B'}:a\notin U'\}$, with inverse $U=U'\cup\{a\}$. Since $|U|=|U'|+1$, we obtain
  \begin{equation}\label{eqn:RC:worm:U:added}
    \sum_{\substack{U\in\Omega_B\\a\in U}}\rho^{|U|}
    =\rho\sum_{\substack{U'\in\Omega_{B'}\\a\notin U'}}\rho^{|U'|}
    \leq\rho Z_{B'}\leq\rho Z_0.
  \end{equation}
  \item If $a\in S$, then $a\notin U$.
  {Dropping the restriction $a\notin U$ enlarges the sum to all of $\Omega_B$, whose total weight is $Z_B$. Hence}
  \begin{equation}\label{eqn:RC:worm:U:removed}
    \sum_{\substack{U\in\Omega_B\\a\notin U}}\rho^{|U|}
    \leq Z_B\leq Z_0.
  \end{equation}
\end{itemize}

Thus, we have
\begin{equation*}
  \begin{aligned}
    T\leq&\; Z_0\sum_{\substack{S\in\Omega_{\leq 2},\,a\in E\setminus S\\S\cup\{a\}\in\Omega_{\leq 2}}}\rho^{|S|+1}\left|f(S\cup\{a\})-f(S)\right| +Z_0\sum_{\substack{S\in\Omega_{\leq 2},\,a\in S\\S\setminus\{a\}\in\Omega_{\leq 2}}}\rho^{|S|}\left|f(S)-f(S\setminus\{a\})\right|\\
    =&\; 2Z_0\sum_{\substack{S\in\Omega_{\leq 2},\,a\in E\setminus S\\S\cup\{a\}\in\Omega_{\leq 2}}}\rho^{|S|+1}\left|f(S\cup\{a\})-f(S)\right|,
  \end{aligned}
\end{equation*}
where the inequality uses \eqref{eqn:RC:worm:U:added} and \eqref{eqn:RC:worm:U:removed}, and the last equality reindexes the second sum by $S\mapsto S\setminus\{a\}$.
This proves \eqref{eqn:RC:worm:covariance}.
\end{proof}

\subsection{Proof of the monotone edge-count Poincar\'e inequality}

Now we are ready to prove \Cref{lem:RC:monotone:Poincare}.

Sample $S\sim\nu_0$ and then lift it through $K(S,\cdot)$ to obtain $X\sim\pi$.
The law of total covariance \eqref{eqn:total:covariance} splits $\operatorname{Cov}(F,|X|)$ into the lift term $\mathbb E[\operatorname{Cov}(F,|X|\mid S)]$ and the even-subgraph term $\operatorname{Cov}(\mathbb E[F\mid S],\mathbb E[|X|\mid S])$.
The first term measures the randomness added when a fixed $S$ is lifted, while the second captures the variation of the even-subgraph sample $S$.
The independent lift decisions control the first term, while \Cref{lem:RC:worm:covariance} bounds the second.

\begin{proof}[Proof of \Cref{lem:RC:monotone:Poincare}]
For $e\in E$ and $A\subseteq E$, define
\begin{equation*}
  \nabla_e f(A):=f(A\cup\{e\})-f(A\setminus\{e\}).
\end{equation*}
Since $f$ is increasing, $\nabla_e f\geq0$.
Moreover, $\nabla_e f(A)=(\nabla_e f)(A\setminus\{e\})$,
so $\nabla_e f(A)$ depends only on $A\setminus\{e\}$.
For every $S\subseteq E$, also define the expected value of $f$ after lifting $S$ by
\begin{equation*}
  \Phi_f(S):=\sum_{A\subseteq E}K(S,A)f(A).
\end{equation*}
The function $\Phi_f$ is increasing: if $S\subseteq T$, couple the two lifts by using the same uniform random number for each edge.
The set lifted from $S$ is then a subset of the set lifted from $T$, and since $f$ is increasing, $\Phi_f(S)\leq\Phi_f(T)$.

Sample $S\sim\nu_0$ and then $X\sim K(S,\cdot)$.
By \Cref{prop:Grimmett:Janson:coupling}, $X\sim\pi$.
Moreover, $\mathbb E[F\mid S]=\Phi_f(S)$ and $\mathbb E[|X|\mid S]=|S|+\rho(m-|S|)=\rho m+(1-\rho)|S|$.
By the law of total covariance \eqref{eqn:total:covariance},
\begin{align}
    \operatorname{Cov}(F,|X|)
    &=\mathbb E\bigl[\operatorname{Cov}(F,|X|\mid S)\bigr]+\operatorname{Cov}\bigl(\mathbb E[F\mid S],\mathbb E[|X|\mid S]\bigr)\notag\\
    &=\mathbb E\bigl[\operatorname{Cov}(F,|X|\mid S)\bigr]+\operatorname{Cov}\bigl(\Phi_f(S),\rho m+(1-\rho)|S|\bigr)\notag\\
    &=\mathbb E\bigl[\operatorname{Cov}(F,|X|\mid S)\bigr]+(1-\rho)\operatorname{Cov}(\Phi_f(S),|S|),\label{eqn:RC:total:covariance}
\end{align}
where the constant term $\rho m$ has zero covariance with $\Phi_f(S)$.
Define
\begin{align}\label{eqn:C-def}
  C_{\mathrm{lift}}:=\mathbb E\bigl[\operatorname{Cov}(F,|X|\mid S)\bigr],\qquad C_{\mathrm{ES}}:=(1-\rho)\operatorname{Cov}(\Phi_f(S),|S|).
\end{align}

Thus, \eqref{eqn:RC:total:covariance} gives $\operatorname{Cov}(F,|X|)=C_{\mathrm{lift}}+C_{\mathrm{ES}}$.
The two terms are the lift and even-subgraph (or subgraph world) contributions, respectively. We bound each term separately.

First, we rewrite the term on the right-hand side of \eqref{eqn:RC:monotone:Poincare}.
Fix $e\in E$ and condition on $X_{-e}=X\setminus\{e\}$, the configuration on the remaining edges.
Let $\theta_e:=\Pr_\pi(e\in X\mid X_{-e})$.
This is the update probability of the Glauber dynamics and, as we discussed before, $\theta_e\in\{\rho,p\}$.
Conditioned on $X_{-e}$, the only randomness comes from the state of $e$.
Notice that
\begin{align*}
  F&=f(X)=f(X_{-e})+\mathbf 1\{e\in X\}\nabla_e f(X),\quad \text{ and } \quad \abs{X}=\abs{X_{-e}}+\mathbf 1\{e\in X\}.
\end{align*}
Then,
\begin{align}
  \operatorname{Cov}(F,|X|\mid X_{-e}) &= \operatorname{Cov}\bigl(f(X_{-e})+\mathbf 1\{e\in X\}\nabla_e f(X), |X_{-e}|+\mathbf 1\{e\in X\}\mid X_{-e}\bigr) \notag\\
  &= \operatorname{Cov}\bigl(\mathbf 1\{e\in X\}\nabla_e f(X), \mathbf 1\{e\in X\}\mid X_{-e}\bigr) \notag\\
  &= \theta_e(1-\theta_e)\nabla_e f(X),\label{eqn:RC:one:edge:covariance}
\end{align}
where the last equality holds because $\mathbf 1\{e\in X\}$ has conditional variance $\theta_e(1-\theta_e)$.

Now we bound the lift term.
Conditional on $S$, write $|X|=|S|+\sum_{e\notin S}\mathbf1\{e\in X\}$ and expand the covariance by linearity.
For each $e\notin S$, the indicator $\mathbf1\{e\in X\}$ is Bernoulli-$\rho$ and independent of $X_{-e}$ under the lift.
Since $F=f(X_{-e})+\mathbf1\{e\in X\}\nabla_e f(X)$ and $\nabla_e f(X)$ depends only on $X_{-e}$, this gives
\begin{equation}\label{eqn:RC:lift:indicator:covariance}
 \operatorname{Cov}\bigl(F,\mathbf1\{e\in X\}\mid S\bigr)
 =\rho(1-\rho)\mathbb E[\nabla_e f(X)\mid S].
\end{equation}
Summing \eqref{eqn:RC:lift:indicator:covariance} over $e\notin S$ and taking expectation over $S$, we obtain
\begin{equation*}
 C_{\mathrm{lift}}
 =\rho(1-\rho)\sum_{e\in E}
   \mathbb E\bigl[\mathbf1\{e\notin S\}\nabla_e f(X)\bigr]
 \leq\rho(1-\rho)\sum_{e\in E}\mathbb E_\pi[\nabla_e f(X)],
\end{equation*}
where the inequality follows from $\nabla_e f\geq0$.

To express this bound in terms of single-edge conditional covariances, we compare $\rho(1-\rho)$ with $\theta_e(1-\theta_e)$.
The heat-bath rule gives $\theta_e\in\{\rho,p\}$ for every $X_{-e}$.
If $\theta_e=\rho$, the two variances are equal.
If $\theta_e=p$, then $\rho=p/(2-p)$ implies
\begin{equation*}
 \frac{\rho(1-\rho)}{\theta_e(1-\theta_e)}
 =\frac{\rho(1-\rho)}{p(1-p)}
 =\frac{2}{(2-p)^2}\leq2.
\end{equation*}
Thus $\rho(1-\rho)\leq2\theta_e(1-\theta_e)$ in both cases.
Applying this comparison inside each expectation, using $\nabla_e f\geq0$, and then substituting \eqref{eqn:RC:one:edge:covariance} gives
\begin{equation}\label{eqn:RC:lift:bound}
 C_{\mathrm{lift}} \leq2\sum_{e\in E}\mathbb E_\pi\bigl[\theta_e(1-\theta_e)\nabla_e f(X)\bigr] =2\sum_{e\in E}\mathbb E_\pi\bigl[\operatorname{Cov}(F,|X|\mid X_{-e})\bigr].
\end{equation}
This matches the right-hand side of \eqref{eqn:RC:monotone:Poincare} with a coefficient of $2$ instead of $2n$.

We next bound the even-subgraph term.
Our main tool will be \eqref{eqn:RC:worm:covariance}.
For $e\notin S$, let $R\sim K(S,\cdot)$. Then $R\cup\{e\}\sim K(S\cup\{e\},\cdot)$.
Under $K(S,\cdot)$, the edge $e$ is added independently with probability $\rho$.
Hence,
\begin{align}
    \Phi_f(S\cup\{e\})-\Phi_f(S)&=\mathbb E\bigl[f(R\cup\{e\})-f(R)\bigr]=\mathbb E\bigl[\mathbf 1\{e\notin R\}\nabla_e f(R)\bigr]\notag\\
    &=\sum_{R\subseteq E:\;e\not\in R}K(S,R)\nabla_e f(R) =(1-\rho)\sum_{A\subseteq E}K(S,A)\nabla_e f(A),\label{eqn:RC:lift:gradient}
\end{align}
where we used the fact that $\nabla_e f(R)$ depends only on $R\setminus \{e\}$.
Since $\Phi_f$ is increasing, substituting \eqref{eqn:RC:lift:gradient} into \eqref{eqn:RC:worm:covariance} and using \eqref{eqn:C-def}, we obtain
\begin{align}
    C_{\mathrm{ES}}&\leq\frac{(1-\rho)^2}{Z_0}\sum_{\substack{S\in\Omega_{\leq 2},\,e\in E\setminus S\\S\cup\{e\}\in\Omega_{\leq 2}}}\rho^{|S|+1}\sum_{A\subseteq E}K(S,A)\nabla_e f(A)\notag\\
    &=\frac{(1-\rho)^2}{Z_0}\sum_{\substack{S\in\Omega_{\leq 2},\,e\in E\setminus S\\S\cup\{e\}\in\Omega_{\leq 2}}}\rho^{|S|}\sum_{\substack{A\subseteq E\\e\in A}}K(S,A)\nabla_e f(A)\notag\\
    &=\frac{(1-\rho)^2}{Z_0}\sum_{e\in E}\sum_{\substack{A\subseteq E\\e\in A}}\nabla_e f(A)\sum_{\substack{S\in\Omega_{\leq 2},\,e\notin S\\S\cup\{e\}\in\Omega_{\leq 2}}}\rho^{|S|}K(S,A).\label{eqn:RC:ES:before:count}
\end{align}

For fixed $e=uv$ and $A\subseteq E$ with $e\in A$, every $S$ in the inner sum of \eqref{eqn:RC:ES:before:count} must have $B:=\partial S$ satisfying $|B|,|B\oplus\{u,v\}|\in\{0,2\}$.
A two-element $B$ disjoint from $\{u,v\}$ is excluded because $|\partial(S \cup \{e\})| = 4$ and hence $S \cup \{e\} \notin \Omega_{\leq 2}$.
Thus, $B$ is $\emptyset$, $\{u,v\}$, $\{u,w\}$, or $\{v,w\}$ for some $w\in V\setminus\{u,v\}$.
There are therefore at most $2+2(n-2)=2n-2$ possible sets $B$.
Grouping the summands by $B=\partial S$ and enlarging each group to all of $\Omega_B$ yields
\begin{equation}\label{eqn:RC:ES:lift:density:bound}
  \begin{aligned}
    \frac{1}{Z_0}\sum_{\substack{S\in\Omega_{\leq 2},\,e\notin S\\S\cup\{e\}\in\Omega_{\leq 2}}}\rho^{|S|}K(S,A)
    &\leq\sum_{\substack{B\subseteq V\\|B|,\,|B\oplus\{u,v\}|\in\{0,2\}}}\frac{1}{Z_0}\sum_{S\in\Omega_B}\rho^{|S|}K(S,A)\\
    &\leq\sum_{\substack{B\subseteq V\\|B|,\,|B\oplus\{u,v\}|\in\{0,2\}}}\pi(A)
    \leq(2n-2)\pi(A).
  \end{aligned}
\end{equation}
The second inequality follows from \eqref{eqn:RC:ES:lift}, and the last from the count above.
Substituting \eqref{eqn:RC:ES:lift:density:bound} into \eqref{eqn:RC:ES:before:count}, 
we have
\begin{align*}
  C_{\mathrm{ES}}
    &\leq(2n-2)(1-\rho)^2\sum_{e\in E}\mathbb E_\pi\bigl[\mathbf 1\{e\in X\}\nabla_e f(X)\bigr].
\end{align*}
Recall that $\nabla_e f(X)$ depends only on $X_{-e}$.
Thus,
\begin{align*}
  \Ex_\pi\bigl[\mathbf 1\{e\in X\}\nabla_e f(X)\bigr] 
  &= \Ex_\pi\bigl[\nabla_e f(X)\Ex_\pi\bigl[\mathbf 1\{e\in X\}\mid X_{-e}\bigr]\bigr]\\
  & = \Ex_\pi\bigl[\theta_e\nabla_e f(X)\bigr].
\end{align*}
For $\theta_e\in\{\rho,p\}$, notice that $(1-\rho)^2\theta_e\leq\theta_e(1-\theta_e)$.

Then, by \eqref{eqn:RC:one:edge:covariance}, we obtain
\begin{equation}\label{eqn:RC:ES:bound}
  \begin{aligned}
    C_{\mathrm{ES}}
    &\leq(2n-2)(1-\rho)^2\sum_{e\in E}\mathbb E_\pi\bigl[\mathbf 1\{e\in X\}\nabla_e f(X)\bigr] =(2n-2)(1-\rho)^2\sum_{e\in E}\mathbb E_\pi\bigl[\theta_e\nabla_e f(X)\bigr]\\
    &\leq(2n-2)\sum_{e\in E}\mathbb E_\pi\bigl[\theta_e(1-\theta_e)\nabla_e f(X)\bigr] =(2n-2)\sum_{e\in E}\mathbb E_\pi\bigl[\operatorname{Cov}(F,|X|\mid X_{-e})\bigr].
  \end{aligned}
\end{equation}
Adding \eqref{eqn:RC:lift:bound} and \eqref{eqn:RC:ES:bound} and using \eqref{eqn:RC:total:covariance} proves \eqref{eqn:RC:monotone:Poincare}.
\end{proof}

\section{Proof of the main theorem}\label{sec:proof:main}
Finally, for completeness, we give the full algorithm and prove the main theorem.

\begin{proof}[Proof of \Cref{thm:main}]
We first apply \Cref{thm:Ising:sparsification} with error $\eps/2$.
This takes $O(m\log^2 n+n)$ time and gives $|V'|\leq n$ and $|E'|=O_\beta(n \log (2n / \eps))$.
If $E'=\emptyset$, since $G$ is connected, $G'$ consists of one vertex; assigning a uniformly random spin to it and applying the function $\Lift$ takes $O(n)$ time.  Hence, we may assume that $E'\neq\emptyset$ and $|V'|\geq 2$.

We next sample from the random-cluster model on $G'$ with $p:=1-1/\beta$ using the Glauber dynamics.  We start from the empty configuration and initialize the fully dynamic connectivity data structure with vertex set $V'$ and no open edge.  Suppose that the current configuration is $X$.
The Glauber dynamics draws a uniformly random edge $e=(u,v)\in E'$, updates $X\leftarrow X\setminus\{e\}$ if $e\in X$, and queries whether $u$ and $v$ are connected in $(V',X)$.  It then updates $X\leftarrow X\cup\{e\}$ with probability $p$ if they are connected and with probability $p/(2-p)$ otherwise.  Consequently, it suffices to support the following operations:
\begin{itemize}
  \item update $X\leftarrow X\cup\{e\}$ for any given $e\in E'$;
  \item update $X\leftarrow X\setminus\{e\}$ for any given $e\in E'$;
  \item query whether $u$ and $v$ are connected in $(V',X)$ for any given $u,v\in V'$.
\end{itemize}
Using the dynamic connectivity data structure in~\cite{WN13,CZ23}, each transition takes $O(\log^2 n)$ amortized time, with initialization taking at most $O((n+|E'|)\log^2 n)$ time.
To handle parallel edges, we maintain the number of open edges between each pair of adjacent vertices and apply the data structure to the underlying simple graph, updating it only when this number changes between zero and positive.
Grouping the edges by their endpoint pairs using a balanced search tree takes $O(|E'|\log n)$ time, and storing a pointer from each edge to its counter makes each subsequent counter update take $O(1)$ time.

By \Cref{thm:RC:Glauber:mixing}, an $\eps/2$-approximate random-cluster sample is obtained after
\[
  t=\left\lceil 2|V'||E'|\log\frac{4|E'|}{\eps}\right\rceil
  =O_\beta\bigl(n^2\log^2(n/\eps)\bigr)
\]
transitions, since $|E'|=O_\beta(n\log(n/\eps))$ and $\log(4|E'|/\eps)=O_\beta(\log(n/\eps))$.
Thus the sampling time, including initialization, is $O_\beta(n^2\log^2 n\,\log^2(n/\eps))$.
Applying the Edwards--Sokal coupling as in \Cref{prop:Edwards:Sokal}, followed by the function $\Lift$, takes $O(n+|E'|)$ time and does not increase the $\eps/2$ sampling error.
Together with the $\eps/2$ sparsification error, this gives the required error bound.
The total running time is therefore $O(m\log^2 n+n)+O_\beta(n^2\log^2 n\,\log^2(n/\eps))$.
If $G$ is simple, then $m=O(n^2)$, so the total running time is $O_\beta(n^2\log^2 n\,\log^2(n/\eps))$.
\end{proof}

\section{Concluding remarks} \label{sec:conclude}

It is natural to wonder what the true mixing time is for the random-cluster dynamics.
Although we believe the upper bound in \Cref{thm:RC:Glauber:mixing} can be improved, we conjecture that it cannot be improved to near-linear time.
The reason is as follows.
First, the closely related Swendsen--Wang dynamics mixes in time $\Theta\left(n^{1/4}\right)$ at the critical temperature for the mean-field case \cite{LNNP14}.
While this does not directly translate into a lower bound for the random-cluster dynamics, it is a strong sign that the RC dynamics does not mix in near-linear time.

Another piece of evidence is the conjectured or predicted behaviour of the Ising model on the 3D torus at the critical temperature.
Basically, the mixing time can be lower bounded by the inverse of the spectral gap, which, in turn, can be bounded using a variational characterisation analogous to that for the log-Sobolev constant.
Taking the size of the open edge set as the test function, the mixing time is lower bounded by the variance of the size of this set.
Using predictions of the critical Ising model on the 3D torus, this variance is super-linear.
See \cite[Section 4.2.1]{Dum22} for some rigorous background and \cite{HP98,KPSDV16} for some numerical results.

\bibliographystyle{amsalpha}
\bibliography{fast-RC}

\appendix

\section{Sparsification for all-terminal network reliability} \label{sec:reliability}

In this section we show that the same contraction sparsification in \Cref{sec:sparsification} also works for all-terminal network reliability.
This improves the previous running time from $\widetilde{O}(mn/\eps^2)$ \cite{CGZZ24} to $\widetilde{O}(m+n^2/\eps^2)$.

Let $G$ be a multigraph.
We may assume that $G$ is connected and loopless.
For $p\in(0,1)$, let $\operatorname{Rel}_p(G)$ denote the probability that $G_p$ is connected, where $G_p$ is the spanning subgraph obtained by deleting every edge independently with probability $p$.

\begin{theorem}[All-terminal reliability sparsification]
  \label{thm:reliability:sparsification}
  For every constant $p\in(0,1)$, there is an algorithm that, given a connected loopless multigraph $G=(V,E)$ and an error bound $\eta\in(0,1)$, constructs a multigraph $H$ via contraction in time $O(m\log^2 n+n)$ such that
  \begin{align*}
    |V(H)|\leq n,
    \quad\text{ and }\quad
    |E(H)|=O_p\bigl(n\log(n/\eta)\bigr),
  \end{align*}
  where $n:=|V|$ and $m:=|E|$, and
  \begin{align*}
    (1-\eta)\operatorname{Rel}_p(H)
    \leq \operatorname{Rel}_p(G)
    \leq \operatorname{Rel}_p(H).
  \end{align*}
  Consequently, for every $\eps\in(0,1)$, there is an FPRAS for $\operatorname{Rel}_p(G)$ with running time $\widetilde O_p(m+n^2/\eps^2)$, where $\widetilde O_p$ hides factors that are polylogarithmic in $n/\eps$ and constants that depend on $p$.
\end{theorem}

\begin{proof}
Apply \Cref{prop:BK:tight:strength} to $G$ to obtain tight strengths $(\widetilde s_e)_{e\in E}$ satisfying \eqref{eqn:strength:estimates}, and let
\begin{align*}
  K:=\left\lceil
    \frac{\log(2n^5/\eta)}{\log(1/p)}
  \right\rceil.
\end{align*}
Let $E_{\geq K}:=\{e\in E:\widetilde s_e\geq K\}$, and let $C_1,\ldots,C_N$ be the connected components of $(V,E_{\geq K})$.
Construct $H$ by contracting each $C_i$ to a single vertex, deleting the resulting self-loops, and retaining every edge between distinct components with its multiplicity.
Every retained edge $e$ satisfies $\widetilde s_e<K$, and hence
\begin{align*}
  |V(H)|=N\leq n,
  \text{ and }
  |E(H)|
  \leq K\sum_{e\in E(H)}\frac{1}{\widetilde s_e}
  \leq K \cdot c_0 n
  =O_p\bigl(n\log(n/\eta)\bigr).
\end{align*}

As in the proof of \Cref{thm:Ising:sparsification}, let $B_1,\ldots,B_L$ be the inclusion-maximal $K$-strong vertex sets of $G$.
These sets partition $V$, and every $C_i$ is contained in a unique $B_\ell$.
Let $X\subseteq E$ be the random set of retained edges, and $G_{\ell}:=(B_\ell,X\cap E(G[B_\ell]))$ be the retained subgraph of $B_\ell$.
Let $\mathcal A$ be the event that $G_{\ell}$ is connected for every $\ell\in[L]$.
For a fixed $B_\ell$ with at least two vertices, the argument in \Cref{thm:Ising:sparsification}, via \Cref{prop:Karger:cut:counting}, gives
\begin{align*}
  \Pr\bigl(G_{\ell}\text{ is disconnected}\bigr)
  &\leq \sum_{j\geq1}n^{2(j+1)} \cdot p^{jK}=\frac{n^4p^K}{1-n^2p^K}.
\end{align*}
The choice of $K$ implies $p^K\leq\eta/(2n^5)$, so the last expression is at most $\eta/n$.
Since there are at most $n$ maximal $K$-strong components, a union bound yields $\Pr(\mathcal A)\geq1-\eta$.

Let $\mathcal C$ be the event that $(V,X)$ is connected, and let $\mathcal D$ be the event that the retained subgraph of $H$ is connected.
In other words, $\operatorname{Rel}_p(G)=\Pr(\mathcal C)$ and $\operatorname{Rel}_p(H)=\Pr(\mathcal D)$.
Contracting the sets $C_i$ shows that $\mathcal C\subseteq\mathcal D$.
Moreover, $\mathcal A\cap\mathcal D\subseteq\mathcal C$: identifying all vertices of $H$ whose corresponding blocks $C_i$ lie in the same $B_\ell$ gives a connected quotient, and given $\mathcal A$ every $B_\ell$ is internally connected in $(V,X)$.
Both $\mathcal A$ and $\mathcal D$ are increasing events under the product measure on $X$.
Therefore, Harris's inequality \cite{Har60} gives
\begin{align*}
  \operatorname{Rel}_p(G)
  &\geq\Pr(\mathcal A\cap\mathcal D)
  \geq\Pr(\mathcal A)\Pr(\mathcal D)
  \geq(1-\eta)\operatorname{Rel}_p(H).
\end{align*}
The inclusion $\mathcal C\subseteq\mathcal D$ gives the reverse inequality $\operatorname{Rel}_p(G)\leq\operatorname{Rel}_p(H)$.

The tight strengths can be computed in $O(m\log^2 n)$ time, and the components and contraction can be constructed in $O(m+n)$ additional time.
This proves the sparsification claim.

For the FPRAS, apply the sparsification with $\eta:=\eps/4$ and run the reliability FPRAS obtained from the near-linear sampler of Chen, Guo, Zhang, and Zou \cite{CGZZ24} on $H$ with relative error $\eps/4$.
If its output is $\widehat R$, then, with its stated success probability,
\begin{align*}
  (1-\eps/4)\operatorname{Rel}_p(H)
  \leq \widehat R
  \leq (1+\eps/4)\operatorname{Rel}_p(H).
\end{align*}
Together with the reliability comparison above, this implies that $\widehat R$ is a relative $\eps$-approximation to $\operatorname{Rel}_p(G)$.
The FPRAS of \cite{CGZZ24} runs in $\widetilde O_p(|V(H)||E(H)|/\eps^2)$ time.
Since $|V(H)|\leq n$ and $|E(H)|=O_p(n\log(n/\eps))$, the total running time, including sparsification, is $\widetilde O_p(m+n^2/\eps^2)$.
\end{proof}

\end{document}